\documentclass[11  pt, oneside,onecolumn]{article} 
\usepackage[margin=1in,lmargin=1in,rmargin=1in,bmargin=1in,tmargin=1in]{geometry}  
\usepackage[latin1]{inputenc}
\usepackage{amsmath}
\usepackage{amsfonts}
\usepackage{amssymb}
\usepackage{caption}
\usepackage{booktabs}
\usepackage{bbm}
\usepackage{blkarray}
\usepackage{circledsteps}

\usepackage{amsfonts}
\usepackage{amsmath}
\usepackage{mathrsfs}  
\usepackage{amssymb}
\usepackage{graphicx}
\usepackage{booktabs}
\usepackage{enumitem}
\usepackage{algorithm}
\usepackage[noend]{algpseudocode}

\usepackage{mathtools}

\usepackage{xparse}

\usepackage{amsthm}
\newtheorem{theorem}{Theorem}
\numberwithin{theorem}{section}
\newtheorem{lemma}[theorem]{Lemma}
\newtheorem{corollary}[theorem]{Corollary}

\numberwithin{observation}{section}

\numberwithin{definition}{section}

\numberwithin{conjecture}{section}

\newtheoremstyle{introstyle}
{}
{}
{\itshape}
{}
{\bfseries}
{.}
{ }
{}

\usepackage{mathrsfs}  

\newcommand{\pr}[1]{\mathbb{P}\left[ #1 \right]}
\newcommand{\prc}[2]{\mathbb{P}\left[ #1 \mid #2 \right]}

\newcommand{\vc}{\textsc{Vertex Cover}}
\newcommand{\cc}{\textsc{Correlation Clustering}}

\usepackage[
pagebackref=true,
colorlinks,
colorlinks=true,
linkcolor=blue,
citecolor = blue
]{hyperref}

\author{Nate Veldt \\ Department of Computer Science and Engineering \\  Texas A\&M University \\ nveldt@tamu.edu}

\begin{document}
	\title{Growing a Random Maximal Independent Set Produces a 2-approximate Vertex Cover}
	\date{}
	\maketitle
	\begin{abstract}
		This paper presents a fast and simple new 2-approximation algorithm for minimum weighted vertex cover. The unweighted version of this algorithm is equivalent to a well-known greedy maximal independent set algorithm. We prove that this independent set algorithm produces a 2-approximate vertex cover, and we provide a principled new way to generalize it to node-weighted graphs. 
		Our analysis is inspired by connections to a clustering objective called correlation clustering. To demonstrate the relationship between these problems, we show how a simple \emph{pivot} algorithm for correlation clustering implicitly approximates a special type of hypergraph vertex cover problem. 
		Finally, we use implicit implementations of this maximal independent set algorithm to develop fast and simple 2-approximation algorithms for certain edge-deletion problems that can be reduced to vertex cover in an approximation preserving way.
	\end{abstract}

	\section{Introduction}
A set of nodes in a graph is a vertex cover if every edge in the graph is adjacent to at least one node in the cover. The \textsc{Vertex Cover} {problem} is the task of finding a minimum cardinality vertex cover in a graph, or a minimum weight cover in the case of node-weighted graphs. This is one of the most well-known NP-hard optimization problems, and the decision version of the problem is one of Karp's 21 NP-complete problems~\cite{karp1972reducibility}. There are many 2-approximation algorithms for both weighted and unweighted \vc{} that date back to the 1970s and 1980s~\cite{bar1981linear,bar1985local,Garey:1990:CIG:574848,pitt1985simple,savage1982depth}. More sophisticated algorithms also exist with approximation factors that are slightly (though not a constant amount) better than 2~\cite{bar1985local,halperin2002improved,karakostas2009better}, while for every constant $\varepsilon > 0$ the problem is UGC-hard to approximate below a factor of $2 - \varepsilon$~\cite{khot2008vertex}.
Independent of the unique games conjecture, the problem is NP-hard to approximate below a factor of 1.3606~\cite{dinur2005hardness}.
\vc{} has also been studied extensively from the perspective of fixed-parameter tractability and kernelization~\cite{abu2006scalable,chen2008kernel,chen2001vertex,chen2006improved,garg2016raising,kratsch2018randomized} and parallel approximation algorithms~\cite{ghaffari2018improved,ghaffari2020massively,koufogiannakis2009distributed}.  Finding a vertex cover is a key substep for many other combinatorial problems and applications~\cite{abu2006scalable,gottlieb2014efficient,park2001effectiveness,sintos2014using,veldt2022correlation}, and many other problems are known to be reducible to or reducible from \vc{} in an approximation preserving way~\cite{garg2004multiway,kenkre2017approximability,klein2016approximability,sintos2014using,veldt2022correlation}. Thus, new algorithmic techniques and hardness results for \vc{} can have far reaching implications for many other problems.

This paper presents a fast and simple 2-approximation algorithm for the minimum weighted \vc{} problem based on growing a maximal independent set. At each iteration, the algorithm samples a node proportional to its weight, adds it to an independent set, and then places all neighboring nodes in the vertex cover. 
This approach highlights several new connections between algorithmic techniques for different problems related to \vc{}. The unweighted version of our algorithm is equivalent to a well-known greedy random method for finding a maximal independent set (MIS), which selects a uniform random ordering of nodes and greedily adds nodes to an independent set~\cite{bennett2016note,blelloch2012greedy,coppersmith1987parallel,fischer2019tight}. 
Although finding maximum independent sets and minimum vertex covers are complementary problems, they are vastly different from the perspective of approximations~\cite{feige1996interactive,zuckerman2006linear}. Furthermore, \emph{maximal} independent sets can be very different from \emph{maximum} independent sets. Despite these differences, our work provides a proof that the greedy random MIS algorithm produces a 2-approximation for \vc{}, and also provides a principled approach for generalizing this MIS algorithm to node-weighted graphs. The analysis of our algorithm also reveals a connection between approximating \vc{} and approximating a problem called \textsc{Correlation Clustering}~\cite{BansalBlumChawla2004}. In particular, the proof of our approximation guarantee is inspired by the analysis of a simple 3-approximation algorithm called \textsc{Pivot}~\cite{ailon2008aggregating}, which we show implicitly approximates a \vc{} problem on a special type of 3-uniform hypergraph. Our results also imply that an existing $O(\log n)$-round parallel algorithm for finding a maximal independent set simultaneously serves as an approximation algorithm for \vc{}, \textsc{Correlation Clustering}, and an edge-labeling problem related to the principle of strong triadic closure~\cite{sintos2014using}.

Finally, we show how to use implicit implementations of our maximal independent set approach to obtain fast and simple approximation algorithms for certain edge-deletion problems that can be reduced to vertex cover in an approximation preserving way. By {implicit implementation}, we mean that the mechanics of our algorithm are applied without forming the reduced instance of \vc{}. For the problems we consider, an implicit implementation of our method can be made asymptotically faster than naively forming the reduced graph or implicitly iterating through all of the edges in the reduced \vc{} instance.
We specifically use our algorithm to develop a simple new combinatorial 2-approximation algorithm for a recent edge-colored hypergraph clustering objective~\cite{amburg2020clustering}, and a faster 2-approximation algorithm for a special case of the \textsc{DAG Edge Deletion} problem~\cite{kenkre2017approximability}.

	\section{Background and Related Work}
Let $G = (V,E)$ denote an undirected graph with $n = |V|$ nodes and $m = |E|$ edges, where each node $v \in V$ is associated with a nonnegative weight $w_v \geq 0$. We use $N(v)$ to denote the set of neighbors of a node $v \in V$. 
When convenient, we will also denote the node set by $V = \{v_1, v_2, \hdots, v_n\}$ and let $w_i$ denote the weight for the $i$th node $v_i$. The goal of the minimum vertex cover problem is to find a set of nodes $S \subseteq E$ that covers all edges and has minimum weight $w(S) = \sum_{v \in S} w_v$.
This can be encoded by the following binary linear program:
	\begin{equation}
	\label{eq:vcblp}
	\begin{array}{ll}
		\min & \displaystyle{\sum_{v\in V}} w_v x_v  \\ 
		\text{s.t.} & x_u + x_v \geq 1 \text{ for $(u,v) \in E$}\\
		& x_v \in \{0,1\} \text{ for $v \in V$}.
	\end{array}
\end{equation}
Before presenting our new algorithm, we survey existing approximation algorithms, previous research on maximal independent sets, and other related work.


\subsection{Approximation Algorithms for Vertex Cover}
The most widely-known 2-approximation algorithm for unweighted \vc{} works by greedily building a maximal matching in $G$ and adding all nodes adjacent to an edge in the matching to a cover. This can be implemented by iterating through edges in an arbitrary order and adding both endpoints of an edge to the cover if the edge can be added to the matching (Algorithm~\ref{alg:match}). This algorithm is attributed both to Gavril and Yannakakis (see~\cite{Garey:1990:CIG:574848} and~\cite{papadimitriou1998combinatorial}). The local-ratio algorithm of Bar-Yehuda and Even~\cite{bar1981linear,bar1985local} (Algorithm~\ref{alg:bye}) can be viewed as a generalization of this algorithm that also works on edge-weighted graphs. Pitt's randomized algorithm~\cite{pitt1985simple} (Algorithm~\ref{alg:pitt}) also iterates through edges, but whenever it encounters an uncovered edge, it samples one of the two endpoints to add to the vertex cover. This strategy is a randomized 2-approximation for weighted \vc{}. Algorithms~\ref{alg:match},~\ref{alg:pitt}, and~\ref{alg:bye} can all be implemented in $O(|E|)$ time. One other way to obtain a 2-approximate \vc{} in the unweighted case is to return the non-leaf nodes of any depth-first tree~\cite{savage1982depth}. This can also be implemented in $O(|E|)$ time, though this method applies only to the unweighted case.

\begin{figure}[t]
	\caption{Algorithms~\ref{alg:match},~\ref{alg:pitt}, and~\ref{alg:bye} are well-known 2-approximation algorithms for \textsc{Vertex Cover}, and Algorithm~\ref{alg:greedymis} is a greedy algorithm for finding a maximal independent set. All run in $O(|E|)$ time. Algorithms~\ref{alg:pitt} and~\ref{alg:bye} apply to node-weighted graphs, whereas Algorithms~\ref{alg:match} and~\ref{alg:greedymis} assume $w_v = 1$ for each $v \in V$. We will prove that \textsc{GreedyMIS} is also a 2-approximation algorithm for unweighted \vc{}, and provide a generalization for node-weighted graphs.}
		\label{fig:all4algs} 
	\begin{minipage}[t]{0.5\linewidth}
		\begin{algorithm}[H]
			\caption{$\textsc{MatchingVC}(G)$}
			\label{alg:match}
			\begin{algorithmic}
				\State $\mathcal{C} \leftarrow \emptyset$  {\tt // Initialize empty cover} 
				\For{$(u,v) \in E$}
				\If{$u \notin \mathcal{C}$ and $v \notin \mathcal{C}$}
				\State {\tt // Add both nodes to cover}
				\State $\mathcal{C} \leftarrow \mathcal{C} \cup \{u,v\}$
				\EndIf
				\EndFor
				\State Return $\mathcal{C}$
			\end{algorithmic}
		\end{algorithm}
		
	\end{minipage} 
	\begin{minipage}[t]{0.5\linewidth}
		\begin{algorithm}[H]
			\caption{$\textsc{PittVC}(G)$}
			\label{alg:pitt}
			\begin{algorithmic}
				\State $\mathcal{C} \leftarrow \emptyset$  {\tt // Initialize empty cover} 
				\For{$(u,v) \in E$}
				\If{$u \notin \mathcal{C}$ and $v \notin \mathcal{C}$}
				\State With prob.\ $\frac{w_v}{w_u + w_v}$: $\mathcal{C} \leftarrow \mathcal{C} \cup \{u\}$
				\State Otherwise: $\mathcal{C} \leftarrow \mathcal{C} \cup \{v\}$
				\EndIf
				\EndFor
				\State Return $\mathcal{C}$
			\end{algorithmic}
		\end{algorithm}

	\end{minipage} 
	\begin{minipage}[b]{0.5\linewidth}
		\begin{algorithm}[H]
	\caption{$\textsc{LocalRatioVC}(G)$}
	\label{alg:bye}
\begin{algorithmic}
	\State $\mathcal{C} \leftarrow \emptyset$  {\tt // Initialize empty cover} 
	\State for each $v \in V$ set $r(v) = w_v$ 
	\For{$(u,v) \in E$}
	\State $M = \min \{r(v), r(u)\}$
	\State $r(v) \leftarrow r(v) - M$
	\State $r(u) \leftarrow r(u) - M$
	\EndFor
	\State $\mathcal{C} = \{v \in V \colon r(v) = 0\}$
	\State Return $\mathcal{C}$
\end{algorithmic}
\end{algorithm}
	\end{minipage}
	\hfill
	\begin{minipage}[b]{0.5\linewidth}
		\begin{algorithm}[H]
			\caption{$\textsc{GreedyMIS}(G)$}
			\label{alg:greedymis}
			\begin{algorithmic}
				\State $\mathcal{I} \leftarrow \emptyset$; $\mathcal{U} \leftarrow V$ 
				\State Generate random uniform node permutation $\sigma$
				\For{$v = v_{\sigma(1)}, v_{\sigma(2)}, \hdots, v_{\sigma(n)}$}
				\If{$v \in \mathcal{U}$}
				\State $\mathcal{I} \leftarrow \mathcal{I} \cup \{v\}$ and $\mathcal{U} \leftarrow \mathcal{U} \backslash \{v \}$
				\For{ $u \in N(v) \cap  \mathcal{U}$}
				\State $ \mathcal{U} \leftarrow  \mathcal{U} \backslash \{u\}$
				\EndFor
				\EndIf
				\EndFor
				\State Return $\mathcal{I}$
			\end{algorithmic}
		\end{algorithm}
	\end{minipage} 
\end{figure}

Several other algorithms achieve a 2-approximation or better for \vc{}, but take longer than $O(|E|)$ time. One approach relies on solving the linear programming (LP) relaxation obtained by replacing the constraint $x_v \in \{0,1\}$ in~\eqref{eq:vcblp} with linear constraints $0 \leq x_v \leq 1$. If $\{x_v^*\}$ denotes an optimal set of dual variables, the set $S = \{v \in V \colon x_v^* \geq 1/2\}$ is a 2-approximate solution for weighted \vc{}. 
Other more sophisticated algorithms have also been developed, including a $2 - \Theta(1/\sqrt{\log n})$ approximation algorithm for a graph with $n$ nodes~\cite{karakostas2009better}, and a $2 - (1 - o(1))\frac{2 \ln \ln \Delta}{\ln \Delta}$  approximation algorithm where $\Delta$ is the maximum degree of the graph~\cite{halperin2002improved}. The latter two algorithms rely on semidefinite programming relaxations. 
In the opposite direction, \emph{list heuristic} algorithms for \vc{}~\cite{avis2007list,delbot2008better} run in $O(|E|)$ time but have approximation factors worse than 2. A list heuristic is an algorithm that iterates through nodes in a fixed order and at each step makes a decision whether to add the current node to the vertex cover or not. These algorithms are designed specifically for unweighted \vc{}. The best known approximation for a list heuristic is $\frac{\sqrt{\Delta}}{2} + \frac32$ where $\Delta$ is the maximum degree~\cite{avis2007list,delbot2008better}.


Among all of these algorithms for \vc{}, Algorithms~\ref{alg:pitt} and~\ref{alg:bye} are unique in that they both achieve a 2-approximation for \emph{weighted} \vc{} in $O(|E|)$ time. Both of these methods rely on iterating through all edges in the graph and deciding whether to add nodes from the edge to the vertex cover. These algorithms could equivalently be described as selecting an arbitrary \emph{uncovered} edge at each iteration, though this requires the algorithm to update the set of covered edges at the end of an iteration. The overall runtime in either case is $O(|E|)$.


\subsection{Finding Maximal Independent Sets}
An independent set in an undirected graph $G = (V,E)$ is a set of nodes in which no two nodes share an edge. Equivalently, a set of nodes $\mathcal{I} \subseteq V$ is an independent set if and only if its complement set $\mathcal{C} = V - \mathcal{I}$ is a vertex cover. The new approximation algorithm we develop for node-weighted \vc{} can in fact be viewed as a generalization of an existing greedy algorithm for finding a maximal independent set (MIS) in an unweighted graph. This \textsc{GreedyMIS} algorithm (Algorithm~\ref{alg:greedymis}) generates a random permutation of nodes and iteratively adds nodes to an independent set. This algorithm and its slight variants date back to roughly the same time period as the earliest \vc{} approximation algorithms~\cite{bennett2016note,coppersmith1987parallel,fischer2019tight,ghaffari2018improved,karp1985fast}. Given the complementary relationship between independent sets and vertex covers, it may at first seem very intuitive to try to approximate \vc{} using a maximal independent set algorithm. However, this simple reasoning overlooks key differences between algorithmic techniques and theoretical guarantees for finding small vertex covers and finding large independent sets. First of all, although finding a \emph{maximum} independent set is equivalent at optimality to finding a minimum vertex cover, these problems are vastly different from the perspective of approximation algorithms, with the former problem being much harder to approximate~\cite{feige1996interactive,zuckerman2006linear}. 
Furthermore, there can be a significant difference between a \emph{maximum} and \emph{maximal} independent set in a graph. As a simple example, consider a star graph on $n$ nodes: the singleton set consisting of the center node in the star is a maximal independent set of size 1, but the the maximum independent set has size $n - 1$.

As a result of these differences, approximating \vc{} and finding a maximal independent set are typically treated as different tasks. Many research papers on finding maximal independent sets do not even mention \vc{}~\cite{alon1986fast,blelloch2012greedy,coppersmith1987parallel,fischer2019tight,karp1985fast,luby1986simple}, while other papers that address both apply different techniques for each problem~\cite{ghaffari2018improved,ghaffari2020massively}. One indirect relationship between these two problems is that any maximal independent set algorithm can be used as a subroutine for approximating \vc{}. If the goal is to approximate \vc{} on a graph $G = (V,E)$, one can first run a MIS algorithm on the \emph{line graph} of $G$. This produces a maximal matching in $G$, which can be combined with Algorithm~\ref{alg:match} to obtain a 2-approximate vertex cover. However, an arbitrary maximal independent set in $G$ provides no guarantees for the \textsc{Vertex Cover} objective in $G$. This can be seen by again considering the maximal independent set consisting of the center node in a star graph.



\subsection{Correlation Clustering and Edge-Deletion Objectives}
Our work builds on connections between \vc{} and \cc{}~\cite{BansalBlumChawla2004}, which is the problem of partitioning an unweighted and undirected graph $G = (V,E)$ into an arbitrary number of clusters in a way that minimizes the number of \emph{mistakes}. There are two types of mistakes: a \emph{positive mistake} is when a pair of adjacent nodes is separated into different clusters, and a \emph{negative mistake} is when two non-adjacent nodes are placed in the same cluster. 
The problem is NP-hard but many approximation algorithms have been developed~\cite{ailon2008aggregating,BansalBlumChawla2004,behnezhad2022almost,cohen2021correlation,cohen2022correlation,ChawlaMakarychevSchrammEtAl2015,veldt2022correlation}. One of the simplest and fastest algorithms is a randomized 3-approximation commonly known as \textsc{Pivot}, which iteratively selects an unclustered node uniformly at random (the \emph{pivot}) and clusters it with all its unclustered neighbors~\cite{ailon2008aggregating}. This is closely related to Algorithm~\ref{alg:greedymis} in that the pivot nodes form a random greedy maximal independent set. This relationship has  also been noted in previous work~\cite{behnezhad2022almost,fischer2019tight}. 

We also draw on connections between \cc{} and an NP-hard edge-labeling problem called \emph{minimum strong triadic closure labeling with edge insertions} (\textsc{MinSTC+}), which is known to be reducible to a special type of hypergraph \vc{} problem~\cite{gruttemeier2022parameterized,gruttemeier2020relation,nuendorf2020,sintos2014using}. Recent work showed how to use \vc{} algorithms as subroutines for \cc{} approximation algorithms~\cite{veldt2022correlation}, though this did not involve new algorithms for the general \vc{} problem. Section~\ref{sec:equiv} expands on these connections between clustering, edge-labeling, and MIS algorithms, and how they relate to our new approximation algorithm for \vc{}. Finally, our algorithmic techniques lead to new approximation algorithms for multiple edge-deletion problems in graphs and hypergraphs, including a recent objective for clustering edge-colored hypergraphs~\cite{amburg2020clustering,veldt2022correlation} and a path-deletion problem in directed acyclic graphs~\cite{kenkre2017approximability}. We cover formal definitions and additional background as needed for these objectives in Section~\ref{sec:faster}.

	\section{The Maximal Independent Set Algorithm for Vertex Cover}

Our main result is a simple algorithm that simultaneously grows a maximal independent set and builds a 2-approximate vertex cover. This algorithm can be seen as a special type of weighted generalization of \textsc{GreedyMIS} (Algorithm~\ref{alg:greedymis}). Our proof that this is a 2-approximation for \vc{} is closely related to the proof that \textsc{Pivot} is a 3-approximation algorithm for \cc{}~\cite{ailon2008aggregating}. We discuss the relationship between these algorithms in more depth in Section~\ref{sec:equiv}.

\begin{algorithm}[t]
	\caption{\textsc{NeighborCover}($G$)}
	\label{alg:neighborcover}
	\begin{algorithmic}[5]
		\State $\mathcal{C} \leftarrow \emptyset$, $\mathcal{I} \leftarrow \emptyset$, $\mathcal{U} \leftarrow V$.
		\While{$\mathcal{U} \neq \emptyset$}
		\State Randomly select $u \in \mathcal{U}$ proportional to $w_u$
		\State $\mathcal{I} \leftarrow \mathcal{I} \cup \{u\}$
		\For{ $v \in N(u) \cap \mathcal{U}$}
		\State $\mathcal{C} \leftarrow \mathcal{C} \cup \{v\}$
		\State $\mathcal{U} \leftarrow \mathcal{U} \backslash \{v\}$
		\EndFor
		\EndWhile
		\State Return $\mathcal{C}$
	\end{algorithmic}
\end{algorithm}

\subsection{Overview and Approximation Guarantee}
Our algorithm for \textsc{Vertex Cover} (Algorithm~\ref{alg:neighborcover}) iteratively grows a cover set $\mathcal{C}$ and an independent set $\mathcal{I}$. During the course of the algorithm, every node that has not yet been added to $\mathcal{I}$ or $\mathcal{C}$ is in an \emph{undecided} node set $\mathcal{U}$. 
At each iteration, the algorithm randomly chooses a node $v \in \mathcal{U}$ proportional to its node weight $w_v$. That node $v$ is added to set $\mathcal{I}$, and all of its undecided neighbors are added to the vertex cover. The algorithm terminates when all nodes are either in $\mathcal{C}$ or $\mathcal{I}$. By design, $\mathcal{I}$ is guaranteed to be a maximal independent set and $\mathcal{C}$ is a vertex cover. We refer to this algorithm as \textsc{NeighborCover}. 

\begin{theorem}
	\label{thm:mainalg}
	\textsc{NeighborCover} is a randomized $2$-approximation algorithm for the minimum weighted \textsc{Vertex-Cover} problem.
\end{theorem}
\begin{proof}
	The linear programming relaxation for \vc{} is given by
	\begin{equation}
		\label{eq:vclp}
		\begin{array}{ll}
			\min & \displaystyle{\sum_{v\in V}} w_v x_v  \\ 
			\text{s.t.} & x_u + x_v \geq 1 \text{ for $(u,v) \in E$}\\
			& x_v \geq 0 \text{ for $v \in V$}
		\end{array}
	\end{equation}
	where there is variable $x_v$ for each node $v \in V$ and a constraint for each edge. The dual of this relaxation the following linear program:
	\begin{equation}
		\label{eq:vclpdual}
		\begin{array}{rl}
			\max & \displaystyle{\sum_{e \in E}} y_e  \\ 
			\text{s.t. for each $u \in V$:} & \sum_{e : u \in e} y_e \leq w_u  \\
			& y_e \geq 0 \text{ for $e \in E$.}
		\end{array}
	\end{equation}
	When $w_v = 1$ for every $v \in E$, the solution to the dual linear program is the largest fractional edge matching. By LP duality theory, every feasible solution to the dual LP is a lower bound for the \vc{} instance. We will show how to construct a feasible solution whose value is half the expected cost of \textsc{NeighborCover}, proving the 2-approximation.
	
	\paragraph{Expected cost of the algorithm.} If $v \in V$ is added to $\mathcal{I}$ by \textsc{NeighborCover}, we will refer to it as a \emph{MIS-node}. If a node $u \in V$ is never chosen as a MIS-node, this means that the algorithm eventually places $u$ in the cover $\mathcal{C}$, incurring a cost of $w_u$. This means that some node $v$ adjacent to $u$ was chosen as a MIS-node in some iteration, so we will charge the cost $w_u$ to the edge $(u,v) \in E$. 
	For an edge $e \in E$, let $A_e$ denote the event that one of the two nodes in $e$ is chosen as a MIS-node in an iteration where both are still undecided, and let $p_e = \pr{A_e}$.
	An edge $e  = (u,v)$ receives a charge if and only if $A_e$ occurs, and it can only receive a charge once. Conditioned on $A_e$ being true, the charge assigned to $e$ depends on whether $u$ or $v$ is chosen as a MIS-node. With probability $w_u/(w_u + w_v)$, node $u$ is chosen as a MIS-node, meaning that node $v$ is placed in the vertex cover and $e$ is charged cost $w_v$. With probability $w_v/(w_u + w_v)$, $u$ is placed in the vertex cover and the charge is $w_u$. If we let $X_e$ be a random variable denoting the charge to edge $e$, then $C = \sum_{e \in E} X_e$ is the total cost incurred by \textsc{NeighborCover} and has the following expected value:
	\begin{align*}
		\mathbb{E}\left[ C \right] &= \sum_{e \in E} \mathbb{E}\left[ X_e \right] = \sum_{e \in E} \mathbb{E}\left[ X_e \mid A_e \right] \pr{A_e} \\
		&=\sum_{e = (u,v) \in E} \left( w_v \cdot \frac{w_u}{w_u+ w_v} + w_u \cdot \frac{w_v}{w_u+ w_v} \right) p_e =  \sum_{e = (u,v) \in E}  \frac{2w_uw_v}{w_u+w_v}p_e.
	\end{align*}
	
	\paragraph{Lower bound.} 
	For a node $u \in V$, let $B_u$ be the event that node $u$ enters the vertex cover $\mathcal{C}$ at some point during the algorithm. For every edge $e = (u,v)$, we have
	\begin{equation}
		\pr{B_u \land A_e} = \prc{B_u}{A_e} \cdot \pr{A_e} = \frac{w_v}{w_u + w_v} \cdot p_e.
	\end{equation}
	Observe now that the node cost $w_u$ can be charged to only one edge $(v,u)$ incident to $u$. This means that for two different edges $e$ and $f$ that share node $u$, the events $B_u \land A_{e}$ and $B_u \land A_{f}$ are disjoint, and more generally we know that for an arbitrary node $u \in V$,
	\begin{equation}
		\label{eq:probsum}
		\sum_{v\colon e = (u,v) \in E} \frac{w_v}{w_u + w_v} p_e = \sum_{e\colon u \in e} \pr{B_u \land A_e} \leq 1.
	\end{equation}
	For each $e = (u,v) \in E$, define a variable $\hat{y}_e = \frac{w_v w_u}{w_u + w_v}p_e$. By~\eqref{eq:probsum}, for every $u \in V$ we have
	\begin{equation}
		\label{eq:constraintsatisfied}
		\sum_{e \colon u \in e} \hat{y}_e = \sum_{v\colon e = (u,v) \in E} \frac{w_v w_u}{w_u + w_v} p_e \leq w_u,
	\end{equation}
	so the variables $\{\hat{y}_e\}_{e \in E}$ satisfy the constraints of the dual LP~\eqref{eq:vclpdual} and we can see that half the expected cost of the algorithm is a lower bound on the \textsc{Vertex-Cover} instance:
	\begin{equation}
		\label{eq:lowerboundvc}
		\sum_{e \in E} \hat{y}_e = \sum_{(u,v) \in E} \frac{w_v w_u}{w_u + w_v}p_e = \frac{1}{2} \mathbb{E}[{C}].
	\end{equation}
\end{proof}
This result immediately implies that \textsc{GreedyMIS} is a 2-approximation for unweighted \vc{}. 
Theorem~\ref{thm:mainalg} can also be viewed as an improved theoretical result for list heuristic algorithms for \vc{}~\cite{avis2007list,delbot2008better}. In particular, the \textsc{ListRight} algorithm~\cite{delbot2008better} is nearly identical to \textsc{GreedyMIS}, and only differs in that the node ordering is given rather than chosen uniformly at random. The best previous approximation factor for this method is $\frac{\sqrt{\Delta}}{2} + \frac32$, where $\Delta$ is the maximum degree, which is obtained by ordering vertices by degree. Our result shows that a random ordering provides an expected 2-approximation. We summarize these observations as a corollary.
\begin{corollary}
	\textsc{GreedyMIS} (Algorithm~\ref{alg:greedymis}) is equivalent to applying \textsc{ListRight}~\cite{delbot2008better} with a uniform random node ordering, and is a randomized 2-approximation for unweighted \vc{}.
\end{corollary}

\subsection{Runtime Guarantees and Implementation}
When the graph $G = (V,E)$ is unweighted, \textsc{NeighborCover} can be implemented by first generating a uniform random permutation to determine the order in which to visit nodes. If the $i$th node that is visited is undecided, it is added to the independent set, otherwise it is a vertex cover node and the algorithm continues to the next step. The random permutation can be generated in $O(|V|)$ time (e.g., using the Fisher-Yates shuffle), so the runtime for the unweighted version is $O(|E|)$. 

If we assume the nodes have arbitrary nonnegative weights, the implementation and runtime analysis is made more challenging by the node sampling procedure. In particular, sampling a node based on its weight from among all undecided nodes is more involved than the random sampling procedure in Pitt's algorithm (Algorithm~\ref{alg:pitt}), which only requires sampling one of two nodes in an edge. A naive sampling procedure would take $O(|V|)$ time each round, which would lead to an overall runtime of $O(|V|^2 + |E|)$, since we must sample a node in each of $O(|V|)$ iterations. With a more careful implementation we can achieve a runtime of $O(|V| \log |V| + |E|)$. One simple way to achieve this is to use the implementation in Algorithm~\ref{alg:nc2}, which decouples the random sampling strategy from the procedure of growing a maximal independent set. Algorithm~\ref{alg:randperm} is used to generate a permutation of \emph{all} nodes based on their weights, and can be implemented in $O(|V| \log |V|)$ time~\cite{wong1980efficient}. In iteration $i$, Algorithm~\ref{alg:nc2} may visit a node that is decided already, but in this case the node will simply be ignored, so that the selection of the next independent set node follows the same sampling distribution.


\begin{algorithm}[t]
	\caption{$\textsc{WeightedShuffle}(\{w_1, w_2, \hdots , w_n\})$}
		\label{alg:randperm}
		\begin{algorithmic}
			\State $U = \{1,2, \hdots , n\}$ 
			\For{$i = 1$ to $n$}
			\State {\tt{Sample $t \in U$ proportional to $w_t$}}
			\State $\sigma(i) = v$; $U \leftarrow U\backslash \{t\}$
			\EndFor
			\State Return $\sigma$
		\end{algorithmic}
\end{algorithm}

\begin{algorithm}[t]
\caption{$\textsc{NeighborCover}(G)$}
\label{alg:nc2}
		\begin{algorithmic}
			\State $\mathcal{C} \leftarrow \emptyset$, $\mathcal{I} \leftarrow \emptyset$
			\State $\sigma = \textsc{WeightedShuffle}(\{w_1, w_2, \hdots , w_n\})$
			\For{$v = v_{\sigma(1)}, v_{\sigma(2)}, \hdots, v_{\sigma(n)}$}
			\State {\tt // Check if $v$ has to be covered}
			\For{$u \in N(v)$}
			\If{$u \in \mathcal{I}$}
			\State $\mathcal{C} \leftarrow \mathcal{C} \cup \{v\}$
			\State {\bf break}
			\EndIf
			\EndFor
			\State {\tt // If not, add $v$ to MIS}
			\If{$v \notin \mathcal{C}$}
			\State $\mathcal{I}  \leftarrow \mathcal{I}  \cup \{v\}$
			\EndIf
			\EndFor		
			\State Return $\mathcal{C}$
		\end{algorithmic}
	\end{algorithm}


The $O(|V| \log |V|)$ term in the runtime is one {disadvantage} of \textsc{NeighborCover} relative to edge-visiting algorithms that run in linear time even in the weighted case (e.g., Algorithms~\ref{alg:pitt} and~\ref{alg:bye}). Nevertheless, the runtime is still $O(|E|)$ when the graph is unweighted and whenever $|E| = \Omega(|V| \log |V|)$. As we shall see later, one advantage of \textsc{NeighborCover} is that it leads to several particularly simple approximation algorithms for certain edge-deletion problems that can be reduced to \vc{} in an approximation preserving way. These implicit implementations can easily be made more efficient than applying a naive approach that relies on explicitly forming the reduced instance of \vc{}.

	\section{Algorithm Equivalence Results}
\label{sec:equiv}
Theorem~\ref{thm:mainalg} shows for the first time that \textsc{GreedyMIS} (Algorithm~\ref{alg:greedymis}) is an expected 2-approximation for unweighted \textsc{Vertex Cover}, and provides a principled new way to generalize this method to node-weighted graphs. Our method is also related to \cc{}~\cite{BansalBlumChawla2004} and strong triadic closure edge-labeling problems~\cite{sintos2014using}. In particular, 
the proof of Theorem~\ref{thm:mainalg} is inspired by the analysis of the 3-approximate \textsc{Pivot} algorithm for \cc{}~\cite{ailon2008aggregating}. In this section we highlight two separate ways in which \textsc{NeighborCover} and \textsc{Pivot} are related. We also discuss a simple existing parallelization scheme for \textsc{GreedyMIS} which, based on our equivalence results, can be viewed as an approximation algorithm for several different problems at once. 

Several of the connections and equivalences highlighted in this section are already present in some form in previous literature, though not all in one place. We bring these connections together to highlight how our new approximation algorithm for \textsc{Vertex Cover} relates to algorithmic techniques for other problems. These connections also lay the groundwork for several open directions for future research that we discuss at the end of the paper.

\subsection{\cc{} and 3-uniform Hypergraph \vc{}}
Many approximation algorithms of \cc{} are based on counting open wedges (also called \emph{bad triangles} or  \emph{bad triplets})~\cite{ailon2008aggregating,BansalBlumChawla2004,behnezhad2022almost,veldt2022correlation}. An open wedge in $G$ is a set of three nodes whose induced subgraph contains only two edges. Every way of clustering these nodes leads to at least one disagreement: either all nodes will be placed in the same cluster (producing a negative mistake) or two adjacent nodes will be separated (producing a positive mistake). 
Letting $\mathcal{W}$ denote the set of open wedges in $G$, the following binary linear program provides a lower bound for the optimal \cc{} objective:
\begin{equation}
	\label{eq:stcblp}
	\begin{array}{ll}
		\min & \displaystyle{\sum_{\{i, j\} \in {V \choose 2}}} x_{ij} \\ 
		\text{s.t.} & x_{ij} + x_{ik} + x_{jk} \geq 1 \text{ for $\{i,j,k\}\in \mathcal{W}$}\\
		& x_{ij} \in \{0,1\} \text{ for $\{i, j\}  \in {V \choose 2}$}.
	\end{array}
\end{equation}
The constraint $x_{ij} + x_{ik} + x_{jk} \geq 1$ reflects that fact that there will be at least one mistake among the node pairs in the open wedge $\{i,j,k\}$. There is a close relationship between this binary program and the binary program for \vc{} in~\eqref{eq:vcblp}. Instead of variables for nodes, there is a variable for each node pair, and instead of the constraint $x_u + x_v \geq 1$ we have $x_{ij} + x_{ik} + x_{jk} \geq 1$. Problem~\eqref{eq:stcblp} in fact encodes a \vc{} problem in a 3-uniform \emph{open wedge hypergraph} $\mathcal{H} = (\mathcal{V}, \mathcal{E})$ constructed from the original graph $G = (V,E)$ as follows:
\begin{itemize}
	\item For each node pair $\{i,j\} \in {V \choose 2}$, define an edge $v_{ij} \in \mathcal{V}$. 
	\item For each open wedge $\{i,j,k\} \subseteq\mathcal{W}$, define a hyperedge $\{v_{ij}, v_{ik}, v_{jk}\} \in \mathcal{E}$.
\end{itemize}
Figure~\ref{fig:cc2vc} provides an illustration of this reduction. Every clustering of nodes in $G = (V,E)$ can be mapped to a vertex cover in $\mathcal{H}$: if the clustering makes a mistake at node pair $\{i,j\}$, this means node $v_{ij}$ is covered in the hypergraph. However, the reverse is not necessarily true, and it is not hard to come up with simple examples where a vertex cover in $\mathcal{H}$ does not translate to a node clustering in $G$ (see Figure~\ref{fig:cc2vc}).
\begin{figure}[t]
	\centering
	\includegraphics[width = .6\linewidth]{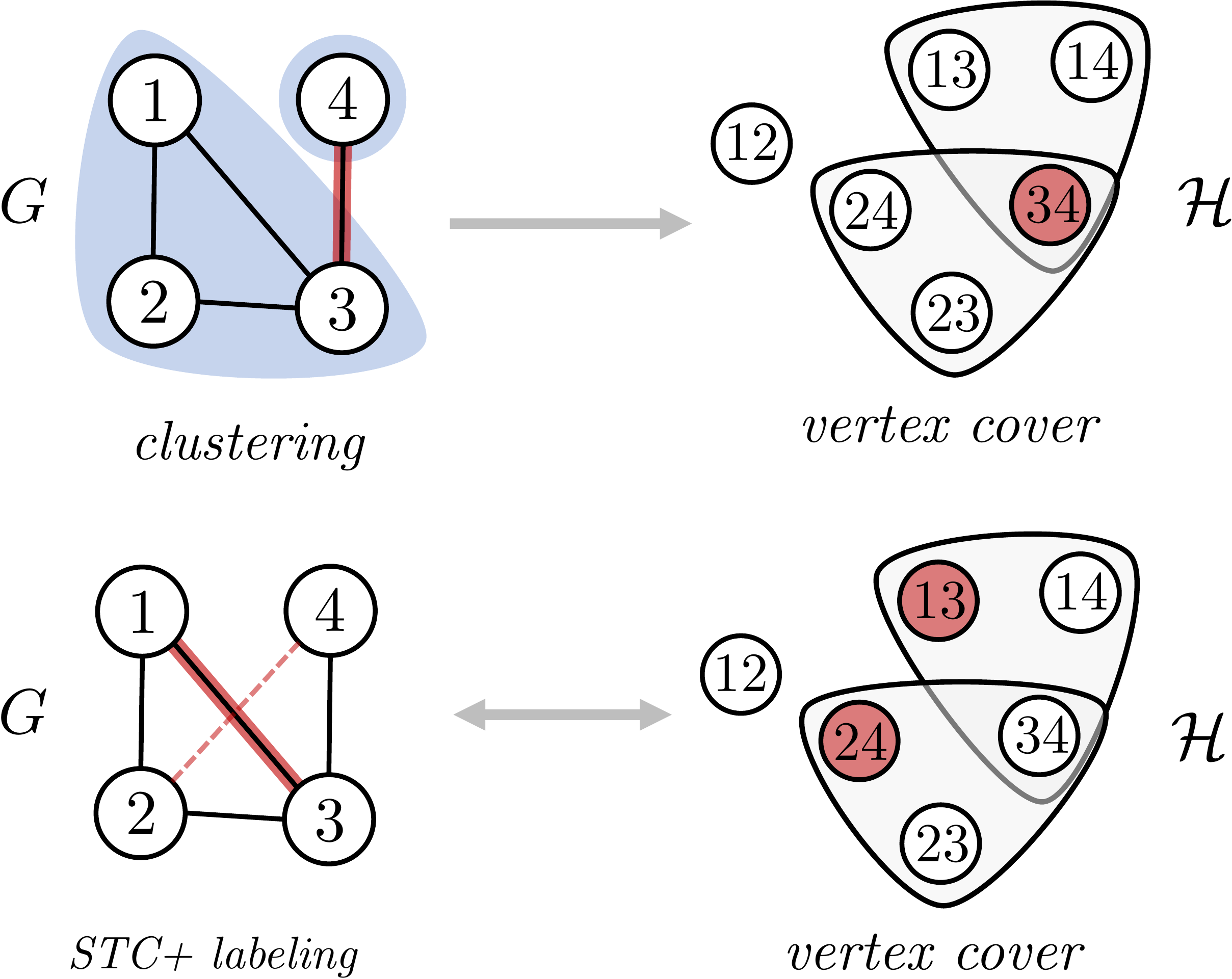}
	\caption{The open wedge hypergraph $\mathcal{H}$ of a graph $G = (V,E)$ is obtained by introducing a node $\Circled{ij}$ for each node pair $\{i,j\}$ in $G$. (Top row) A clustering of the nodes in $G$ always maps to a vertex cover in $\mathcal{H}$. In the example, the only ``mistake'' is to separate nodes 3 and 4, hence node $\Circled{34}$ is a vertex cover in $\mathcal{H}$. (Bottom row) Vertex covers in $\mathcal{H}$ do not always correspond to clusterings, but are in one-to-one correspondence with STC+ labelings. Covering nodes $\Circled{13}$ and $\Circled{24}$ means labeling $(1,3)$ as \emph{weak} and introducing a new (\emph{weak}) edge $(2,4)$ in $G$. All open wedges in $G$ now satisfy strong triadic closure (at least one edge in each open wedge is \emph{weak}).}
	\label{fig:cc2vc}
\end{figure}

This special 3-uniform \vc{} problem 
 is equivalent to an NP-hard edge \emph{labeling} problem that is based on the principle of \emph{strong triadic closure}~\cite{sintos2014using,veldt2022correlation}.
Strong triadic closure (STC) posits that two individuals in a social network will share at least a weak connection to one another if they both share strong ties to a mutual friend (see chapter 3 in~\cite{easley2010networks}). This principle can be related back to open wedges in the graph $G = (V,E)$. If $\{u,v,w\} \in \mathcal{W}$ and $(v, w) \notin E$, this means that $v$ and $w$ have a mutual ``friend'' (node $u$) though they do not share an edge. Strong triadic closure suggests that one of the following must be true: (1) $(u,v)$ is a weak tie, (2) $(u,w)$ is a weak tie, or (3) nodes $v$ and $w$ actually do share at least a weak tie but the graph $G$ simply has a ``missing'' edge. An \textsc{STC+} labeling\footnote{The `+' in STC+ indicates that edge additions are allowed. Sintos and Tsasparas~\cite{sintos2014using} also considered a version that only involved labeling existing edges as weak or strong.} for graph $G = (V,E)$ is defined to be a set of edges $E_W \subseteq E$ to label as \emph{weak} along with a set of node pairs $E_N \subseteq {V \choose 2} - E$ to turn into new \emph{weak} edges, in order to ensure that strong triadic closure holds. In other words, for an open wedge $\{u,v,w\}$ centered at $u$, either $(u,v)$ or $(u,w)$ is labeled as weak, or the non-adjacent pair $\{v,w\}$ is added to $E_N$. The \textsc{MinSTC+} problem is the task of finding an \textsc{STC+} labeling that minimizes $|E_N| + |E_W|$.
 The equivalence between \textsc{MinSTC+} and a special type of 3-uniform  \vc{} problem was noted when this edge-labeling problem was first introduced~\cite{sintos2014using}. Further connections between \textsc{MinSTC+} and \cc{} were explored in subsequent work~\cite{gruttemeier2022parameterized,gruttemeier2020relation,nuendorf2020, veldt2022correlation}. These connections provide the foundation for understanding the relationship between \textsc{NeighborCover} and the \textsc{Pivot} approximation algorithm for \cc{}.

\subsection{\textsc{Pivot} as a Hypergraph Vertex Cover Algorithm}
The \textsc{Pivot} algorithm for \cc{} selects an unclustered node uniformly at random (the \emph{pivot} node) in each iteration, and clusters it with all its unclustered neighbors. This is repeated until all nodes are clustered. Ailon, Charikar, and Newman~\cite{ailon2008aggregating} proved that this algorithm provides a 3-approximation for \cc{} by considering the linear programming relaxation of objective~\eqref{eq:stcblp}. These authors showed that the expected cost of \textsc{Pivot} can be bounded below by constructing an implicit feasible solution for the dual linear program, which encodes the notion of a fractional open wedge packing. An open wedge packing is a node-pair-disjoint set of open wedges in $G$, which provides a lower bound for \cc{} since at least one mistake must be made at each disjoint open wedge. The dual LP encodes \emph{fractional} packings in the sense that each node pair is allowed to partially contribute to multiple open wedges as long as the sum of contributions is at most 1.

The connection between \cc{} and 3-uniform \vc{} was not explicitly noted in the work of Ailon, Charikar, and Newman~\cite{ailon2008aggregating}, but this relationship sheds light on why the analysis for \textsc{Pivot} can be adapted to prove \textsc{NeighborCover} is a 2-approximation for \vc{}. In particular, the fractional open wedge packing that \textsc{Pivot} relies on corresponds to a fractional matching in the open wedge hypergraph, just as \textsc{NeighborCover} relies on a fractional matching lower bound in a graph. We formalize the relationship with a simple lemma that follows quickly from previous observations, but has not been explicitly noted elsewhere in the literature.
\begin{lemma}
	\label{lem:pivot}
	\textsc{Pivot} is a 3-approximation algorithm for \textsc{MinSTC+}. Equivalently, \textsc{Pivot} is a 3-approximation algorithm for the problem of finding a minimum vertex cover in the open wedge hypergraph of a graph.
\end{lemma}
\begin{proof}
	The original analysis of \textsc{Pivot}~\cite{ailon2008aggregating} shows that the expected cost of this algorithm is at most 3 times the optimal solution value of the linear programming relaxation of objective~\eqref{eq:stcblp}. This linear program lower bounds \textsc{MinSTC+} in addition to lower bounding \cc{}. Because every clustering of $G$ also maps to a vertex cover in its open wedge hypergraph (i.e., a valid STC+ labeling), \textsc{Pivot} returns an edge-labeling that is a 3-approximation for \textsc{MinSTC+}.
\end{proof}
The fact that \textsc{Pivot} is a 3-approximation for both \cc{} and \textsc{MinSTC+} is somewhat surprising given the difference between these problems. As mentioned previously, every clustering of $G$ can be mapped to a vertex cover in the open wedge hypergraph $\mathcal{H}$, but the reverse statement is not true. It was recently shown that any $\alpha$-approximation for \textsc{MinSTC+} can be used to design a $(2\alpha)$-approximation for \cc{}~\cite{veldt2022correlation}. This procedure starts with an $\alpha$-approximate vertex cover in  $\mathcal{H}$ and then applies a rounding step that distorts the approximation by a factor of 2 in order to convert the vertex cover in $\mathcal{H}$ into a clustering in $G$. With this result in hand, one can also prove that any $\alpha$-approximation for \cc{} can provide a $(2\alpha)$-approximation for \textsc{MinSTC+}. However, Lemma~\ref{lem:pivot} indicates that \textsc{Pivot} is able to overcome this factor $2$ difference. 

While Lemma~\ref{lem:pivot} provides insight into one relationship between \textsc{Pivot} and \textsc{NeighborCover}, there are still a few key differences between how these algorithms apply to \vc{} problems. First of all, \textsc{Pivot} applies to a very specific type of 3-uniform hypergraph, and even then only implicitly. Its analysis provides no guarantees for the general 3-uniform hypergraph \vc{} problem, while \textsc{NeighborCover} applies to all graph \vc{} problems. Secondly, \textsc{NeighborCover} applies to node-weighted \vc{}, whereas \textsc{Pivot} does not apply to weighted \cc{}. Using a weighted shuffling procedure such as Algorithm~\ref{alg:randperm} to choose pivot nodes does not make sense in the context of \cc{}, since weighted versions of \cc{} involve \emph{edge} weights and not node weights. Finally, perhaps the most interesting difference is that a single iteration of \textsc{Pivot} (implicitly) adds multiple nodes in $\mathcal{H}$ to an independent set. This is because clustering a pivot node $v \in V$ with its neighbors in $G$ means \emph{not} making a mistake at all node pairs involving $v$. In other words, the nodes in $\mathcal{H}$ corresponding to multiple node pairs in $G$ will not be added to the implicit vertex cover. As an example, the clustering of graph $G$ in Figure~\ref{fig:cc2vc} can be obtained by selecting nodes $2$ and $4$ as pivots, in that order. When $2$ is selected as a pivot, all nodes in $\mathcal{H}$ other than the node corresponding to edge $(3,4) \in E$ are added to an independent set in $\mathcal{H}$. In contrast, \textsc{NeighborCover} adds a single node to an independent set in each iteration. 

\subsection{Equivalence among Pivot, GreedyMIS, and NeighborCover}
Although \textsc{NeighborCover} grows an independent set in $G$ in a different way than \textsc{Pivot} grows an independent set in the open wedge hypergraph of $G$, this is essentially because the algorithms are actually very similar in a different regard. Namely, they both operate \emph{on the graph $G$} by selecting a random undecided node in each iteration and making a decision about how to deal with that node's undecided neighbors. Here, \emph{undecided} either means \emph{unclustered} (in the case of \textsc{Pivot}) or \emph{not assigned to a cover or independent set} (in the case of \textsc{NeighborCover}).
\begin{algorithm}[t]
	\caption{\textsc{Pivot}($G$)}
	\label{alg:pivot}
	\begin{algorithmic}[5]
		\State $\mathcal{U} = V$ \hspace{2cm} {\tt // unclustered node set}
		\State For $v \in V$, $c[v] = 0$ {\tt // initialize cluster indicator vector}
		\State $\mathit{clus} = 1$\hspace{1.8cm}  {\tt // current cluster index}
		\State Generate random uniform node permutation $\sigma$
		\For{$v = v_{\sigma(1)}, v_{\sigma(2)}, \hdots, v_{\sigma(n)}$}
		\If{$v \in \mathcal{U}$}
		\State $c[v] = \mathit{clus}$
		\State $\mathcal{U} \leftarrow  \mathcal{U} \backslash \{v\}$
		\For{ $u \in N(v) \cap  \mathcal{U}$}
		\State $c[u] = \mathit{clus}$
		\State $\mathcal{U} \leftarrow  \mathcal{U} \backslash \{u\}$
		\EndFor
		\State $\mathit{clus} \leftarrow \mathit{clus} + 1$
		\EndIf
		\EndFor
		\State Return $c$
	\end{algorithmic}
\end{algorithm}
Pseudocode for \textsc{Pivot} is given in Algorithm~\ref{alg:pivot}, written in a way that best highlights its close relationship to the unweighted version of \textsc{NeighborCover}, i.e., \textsc{GreedyMIS}. In particular, the set of {pivot} nodes defining its clusters exactly corresponds to a maximal independent set grown from a uniform random ordering of nodes. This relationship between \textsc{Pivot} and \textsc{GreedyMIS} has already been noted in previous work~\cite{behnezhad2022almost,fischer2019tight}. Combining this observation with Theorem~\ref{thm:mainalg} and Lemma~\ref{lem:pivot} leads to the following simple corollary.
\begin{corollary}
	\label{cor:3problems}
	Running \textsc{GreedyMIS} on a graph $G$ simultaneously produces a maximal independent set in $G$, an expected 2-approximate \vc{} for $G$, and pivot nodes for an expected 3-approximation for \cc{} and 3-approximation for \textsc{MinSTC+} on $G$.
\end{corollary} 

One interesting consequence of this corollary is that a parallel variant of \textsc{GreedyMIS} also directly provides a simple parallel approximation algorithm for \vc{}. This parallel variant generates a uniform random ordering of nodes, and in each round, all nodes that come before their neighbors in the ordering are added to the independent set. These nodes and their neighbors are removed from the graph, and the algorithm repeats this procedure in rounds until no nodes are left~\cite{blelloch2012greedy,fischer2019tight}. For a fixed ordering of nodes, this returns the same output as the sequential \textsc{GreedyMIS} algorithm, and with high probability the algorithm requires only $O(\log n)$ rounds~\cite{fischer2019tight} before termination. Our equivalence result implies this is a parallel $O(\log n)$-round 2-approximation for  \vc{} as well. 

	\section{Fast and Simple Algorithms for Edge-Deletion Problems}
\label{sec:faster}
\textsc{NeighborCover} can be used to design fast and simple approximation algorithms for combinatorial problems that can be reduced to \vc{}. In particular, we consider certain edge-deletion problems whose reduction to \vc{} leads to a graph with a very special edge structure. By implicitly implementing \textsc{NeighborCover} and taking advantage of this special edge structure, we can design methods that are significantly faster than forming the \vc{} instance explicitly and applying a linear-time \vc{} algorithm as a black-box. 

For the problems we consider, a careful implicit implementation of edge-visiting algorithms for \vc{} (e.g., Algorithms~\ref{alg:pitt} and~\ref{alg:bye}) can also lead to improvements over forming the reduced graph explicitly. However, one advantage of \textsc{NeighborCover} is that it is often simpler to implicitly iterate through the nodes of the reduced graph in an efficient way than to implicitly iterate through the edges. If \emph{all} edges in the reduced graph are visited implicitly, this leads to the same runtime issues as explicitly forming the \vc{} instance, so one must carefully reason about edges that can be safely skipped. We avoid this issue altogether when implicitly implementing \textsc{NeighborCover}, which iterates over nodes in the reduced graph.

\subsection{Minimum Delete-to-Matching}
As an illustrative warm-up, we consider a simple edge-deletion problem where the goal is to find a minimum weight set of edges in an edge-weighted undirected graph $\mathcal{G} = (\mathcal{V} ,\mathcal{E})$ to delete in order to convert $\mathcal{G}$ into a matching. Let $e_i \in \mathcal{E}$ denote the $i$th edge (for an arbitrary ordering of edges), and let $\omega_i \geq 0$ be its weight. We refer to this as \emph{minimum delete-to-matching}, or simply \textsc{MinD2M}. 
This problem can be optimally solved in polynomial time by computing a maximum matching and deleting all edges not in the matching. We will illustrate how to obtain a much faster 2-approximation algorithm using an implicit implementation of \textsc{NeighborCover}, which amounts to finding a maximal matching in the edge-weighted graph $\mathcal{G}$ using a specific edge-sampling strategy. 

\begin{figure}[t]
	\centering
	\includegraphics[width = .85\linewidth]{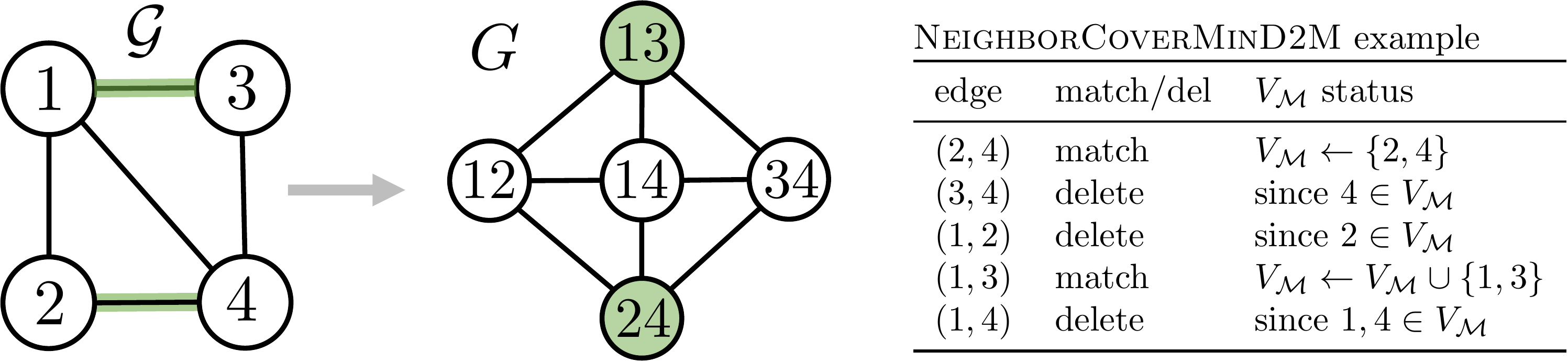}
	\caption{A small example of converting an (unweighted) \textsc{MinD2M} problem on a graph $\mathcal{G}$ into a \vc{} problem on graph $G$ and implicitly applying \textsc{NeighborCover} for the given ordering of edges in $\mathcal{G}$. Maximal independent set nodes are in green, and correspond to a matching in $\mathcal{G}$. An implicit implementation of \textsc{NeighborCover} does not need to check all edges in $G$, but can tell when an edge $(i,j)$ in $\mathcal{G}$ (i.e., node in $G$) is ineligible to be added to a matching in $\mathcal{G}$ (i.e., independent set in $G$) based on whether nodes $i$ and $j$ are already in the matching.}
	\label{fig:mind2m}
\end{figure}
\begin{algorithm}[t]
	\caption{$\textsc{NeighborCoverMinD2M}(G)$}
	\label{alg:misd2m}
	\begin{algorithmic}
		\State \textbf{Input:} edge-weighted graph $\mathcal{G} = (\mathcal{V} ,\mathcal{E})$
		\State \textbf{Output:} 2-approximate solution to \textsc{MinD2M}
		\State $\mathcal{D} \leftarrow \emptyset$  {\tt // edges to delete}
		\State $\mathcal{M} \leftarrow \emptyset$ {\tt // edges in matching}
		\State $\mathcal{V}_\mathcal{M} = \emptyset$ {\tt // nodes in the matched edges}
		\State $\sigma = \textsc{WeightedShuffle}(\{\omega_1, \omega_2, \hdots , \omega_{|\mathcal{E}|}\})$
		\For{$e = e_{\sigma(1)}, e_{\sigma(2)}, \hdots, e_{\sigma(|\mathcal{E}|)}$}
		\State $e = (i,j) $  { \tt // identify nodes in $e$}
		\If{$i \in \mathcal{V}_\mathcal{M}$ or $j \in \mathcal{V}_\mathcal{M}$}
		\State $\mathcal{D} \leftarrow \mathcal{D} \cup \{e\}$
		\Else
		\State $\mathcal{M} \leftarrow \mathcal{M} \cup \{e\}$
		\State $ \mathcal{V}_\mathcal{M} \leftarrow \mathcal{V}_\mathcal{M} \cup \{i, j\}$
		\EndIf
		\EndFor
		\State Return $\mathcal{D}$
	\end{algorithmic}
\end{algorithm}

The \textsc{MinD2M} objective on $\mathcal{G}$ is equivalent to solving \vc{} on the line graph $G = (V,E)$ of $\mathcal{G}$: a node in $\mathcal{G}$ is an edge in $G$, and a maximal independent set in $\mathcal{G}$ is a maximal matching in $G$. Explicitly forming the line graph of $\mathcal{G}$ and applying any linear-time \vc{} algorithm as a black-box yields a 2-approximation algorithm for \textsc{MinD2M}.\footnote{One must be careful to distinguish this from a reduction that has often been applied other direction, namely, approximating \emph{unweighted} \vc{} by first obtaining a maximal matching, which itself can be found by running a MIS algorithm on the line graph. Here, instead of using an \emph{unweighted} maximal matching algorithm to approximate \emph{unweighted} \vc{}, were are using a \emph{weighted} \vc{} approximation algorithm to approximate (the minimization version of) a \emph{weighted} matching problem.}
This basic approach takes $O(|\mathcal{E}|^2)$ time, as this is a bound on the number of edges in $G$. Even if we avoid forming $G$ explicitly, \vc{} algorithms that iterate through all of the edges in $G$ will take $O(|\mathcal{E}|^2)$ time, even if they only implicitly visit edges in $G$ by iterating through pairs of adjacent edges in $\mathcal{G}$. In contrast, Algorithm~\ref{alg:misd2m} is an implicit implementation of \textsc{NeighborCover} applied to \textsc{MinD2M} that takes $O(|\mathcal{E}| \log |\mathcal{E}|)$ time when $\mathcal{G}$ has arbitrary edge weights, and $O(|\mathcal{E}|)$ time in the unweighted case. 
At each iteration, the MIS algorithm must check whether a node in $G$ (i.e., an edge in $\mathcal{G}$) can be added to an independent node set in $G$ (i.e., a matching $\mathcal{M}$ in $\mathcal{G}$). The key to a fast implementation is realizing that we can quickly see if an edge can be added to $\mathcal{M}$ by checking whether either of its nodes already belongs to an edge in $\mathcal{M}$. Selecting a random permutation of edges takes $O(|\mathcal{E}| \log |\mathcal{E}|)$ in the weighted case or only $O(|\mathcal{E}|)$ in the unweighted case. The rest of the algorithm takes $O(|\mathcal{E}|)$ time, since each iteration just involves visiting an edge $(i,j) \in \mathcal{E}$, checking if $i$ or $j$ already belongs to a matched node set $\mathcal{V}_\mathcal{M}$, and then either adding $(i,j)$ to the matching or deleting it. See Figure~\ref{fig:mind2m} for an illustration of running Algorithm~\ref{alg:misd2m}. To summarize, keeping track of one additional fact about each node in $\mathcal{G}$ (specifically, whether or not it is in a set of matched nodes $\mathcal{V}_\mathcal{M}$) is sufficient to avoid iterating through all edges in the reduced \vc{} instance. 

\subsection{DAG Edge Deletion}
Let $D = (\mathcal{V}, A)$ be an edge-weighted directed acyclic graph where $e_i \in A$ is the $i$th directed edge and $\omega_i \geq 0$ is its weight. The DAG \textsc{Edge Deletion} problem with parameter $k$, or simply $\textsc{Ded}$-$k$, seeks a minimum weight set of edges to remove in order to destroy all paths of length $k$ in $D$. The problem was first considered by Kendre et al.~\cite{kenkre2017approximability} as the minimization version of the \textsc{Max}-$k$-\textsc{Ordering} problem. It is the edge-deletion version of the DAG \textsc{Vertex Deletion} problem~\cite{paik1994deleting}. 

We focus on $\textsc{Ded}$-$2$ specifically. This can be reduced in an approximation preserving way to an instance of \vc{} on a graph $G = (V,E)$ by replacing each directed edge $e$ in $D$ with a node $v_e$ in $G$, and by adding an edge between two nodes $v_e$ and $v_f$ in $G$ when the edges $\{e,f\}$ define a directed path in $D$. $\textsc{Ded}$-$2$ is known to be NP-hard, as can be observed from its equivalence (at optimality) with the maximum directed cut problem~\cite{klein2016approximability,lampis2008algorithmic}. Kendre et al.~\cite{kenkre2017approximability} presented two combinatorial 2-approximation algorithms for this problem, one for unweighted graphs and another for weighted graphs. For our purposes it is interesting to note that the unweighted algorithm corresponds to an implicit implementation of the maximal matching method (Algorithm~\ref{alg:match})---while there exists a directed 2-path in the graph, find it and delete both edges. The weighted algorithm is similarly an implicit implementation of the local ratio method (Algorithm~\ref{alg:bye}). Kendre et al.~\cite{kenkre2017approximability} confirmed that the algorithms run in polynomial time, but did not provide any strategies for quickly finding paths in the directed graph that need to be covered. The number of length two paths can be significantly larger than $O(|A|)$, so iterating through all of these paths can be much worse than linear-time in terms of the size of $D$, even in the unweighted case.
\begin{figure}[t]
	\centering
	\includegraphics[width = .85\linewidth]{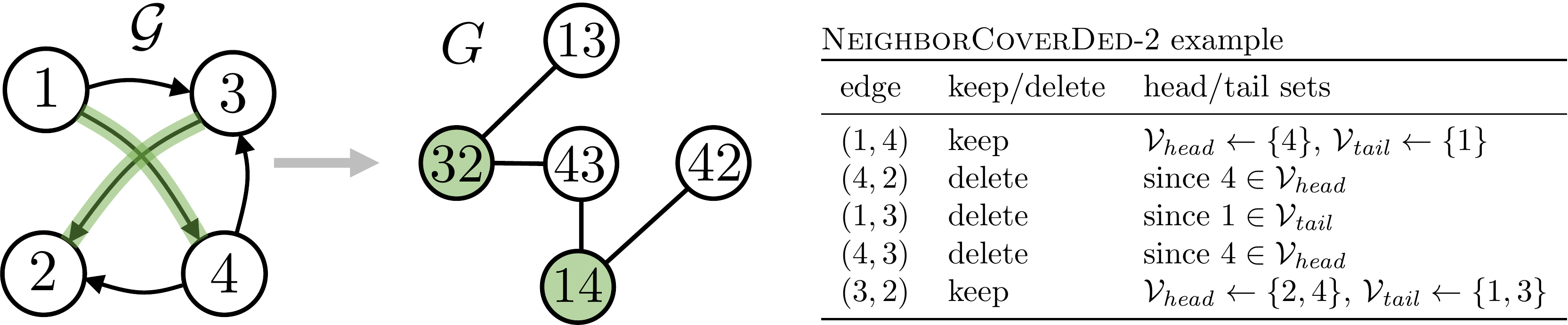}
	\caption{A small example of converting an instance of \textsc{Ded}-$2$ on $\mathcal{G}$ into a \vc{} problem on graph $G$ and implicitly applying \textsc{NeighborCover} for the given ordering of edges in $\mathcal{G}$ (nodes in $G$). The algorithm does not need to check all edges in $G$. It takes constant time to check whether endpoints of a directed edge $(i,j)$ are in $\mathcal{V}_\mathit{head}$ and $\mathcal{V}_\mathit{tail}$, and this is sufficient to know whether node $\Circled{ij}$ can be added to an independent set in $G$.}
	\label{fig:ded2}
\end{figure}
\begin{algorithm}
	\caption{$\textsc{NeighborCoverDed2}(G)$}
	\label{alg:ncded2}
	\begin{algorithmic}
		\State \textbf{Input:} edge-weighted DAG $D = (\mathcal{V},A)$
		\State \textbf{Output:} 2-approximate solution to \textsc{Ded}-$2$
		\State $\mathcal{D} \leftarrow \emptyset$  {\tt // edges to delete}
		\State $\mathcal{K} \leftarrow \emptyset$ {\tt // edges to keep}
		\State $\mathcal{V}_\mathit{head} = \emptyset$ {\tt // head nodes for edges in $\mathcal{K}$}
		\State $\mathcal{V}_\mathit{tail} = \emptyset$ {\tt // tail nodes for edges in $\mathcal{K}$}
		\State $\sigma = \textsc{WeightedShuffle}(\{\omega_1, \omega_2, \hdots , \omega_{|A|}\})$
		\For{$e = e_{\sigma(1)}, e_{\sigma(2)}, \hdots, e_{\sigma(|A|)}$}
		\State $e = (i,j) $  { \tt // identify nodes in $e$}
		\If{$i \in \mathcal{V}_\mathit{head}$ or $j \in \mathcal{V}_\mathit{tail}$}
		\State $\mathcal{D} \leftarrow \mathcal{D} \cup \{e\}$
		\Else
		\State $\mathcal{K} \leftarrow \mathcal{K} \cup \{e\}$
		\State $ \mathcal{V}_\mathit{tail} \leftarrow  \mathcal{V}_\mathit{tail}  \cup \{i\}$
		\State $ \mathcal{V}_\mathit{head} \leftarrow  \mathcal{V}_\mathit{head}  \cup \{j\}$
		\EndIf
		\EndFor
		\State Return $\mathcal{D}$
	\end{algorithmic}
\end{algorithm}

Algorithm~\ref{alg:ncded2} is an implicit implementation of \textsc{NeighborCover} that gives a 2-approximation for \textsc{Ded}-$2$ in time $O(|A|)$ for unweighted graphs and $O(|A| \log |A|)$ time for the weighted case. This method builds an independent set and a vertex cover in $G$ implicitly by building a set $\mathcal{K}$ of edges to keep and a set $\mathcal{D}$ of edges to delete in the acyclic graph $D$. The algorithm searches through directed edges in $D$ (i.e., nodes in $G$), in search of edges that can be added to the set $\mathcal{K}$ without creating 2-paths. Similar to our observations for \textsc{MinD2M}, there is an easy way to check whether an edge is ``allowed'' to be added to $\mathcal{K}$. Given an edge $(i,j)$ where $i$ is the tail node and $j$ is the head node, we know we can add $(i,j)$ to $\mathcal{K}$ as long as there is currently no edge in $\mathcal{K}$ where $i$ is the head or $j$ is the tail. This can be checked in constant time in each iteration. See Figure~\ref{fig:ded2} for an illustration of this process.

\subsection{Edge-colored Hypergraph Clustering}
We finally present a simple new approximation algorithm for the \textsc{Colored Edge Clustering} problem in hypergraphs~\cite{amburg2020clustering}. The input is a hypergraph $\mathcal{H} = (\mathcal{V}, \mathcal{E}, \ell, k)$ where $w_e \geq 0$ denotes the weight of a hyperedge $e \in \mathcal{E}$ and $\ell \colon \mathcal{E} \rightarrow \{1, 2, \hdots , k\}$ maps each hyperedge to one of $k$ colors. The goal is to construct a node color label function $Y \colon V \rightarrow \{1,2, \hdots , k\}$ that disagrees as little as possible with the hyperedge colors. We say a hyperedge $e$ is \emph{satisfied} if $Y[u] = \ell(e)$ for every $e \in u$. Formally, the goal is to minimize the weight of unsatisfied hyperedges. This is equivalent to deleting a minimum weight set of hyperedges so that remaining hyperedges of different colors never overlap.
A node labeling can be viewed as a partitioning of nodes into $k$ clusters where each cluster corresponds to one color. The problem is known to be APX-hard, but various approximation algorithms have been designed~\cite{amburg2020clustering,veldt2022correlation}. 


The best approximation factors for \textsc{Colored Edge Clustering}  are based on linear programming~\cite{amburg2020clustering,veldt2022correlation}, but faster 2-approximations are obtained by reducing \textsc{Colored Edge Clustering} to \vc{}~\cite{veldt2022optimal}. For this reduction, each hyperedge $e \in \mathcal{E}$ corresponds to a node $v_e$ with node-weight $w_e$ in a new graph $G = (V,E)$. Two nodes in $G$ share an edge if they correspond to hyperedges in $\mathcal{H}$ that overlap and have different colors. A naive approach that explicitly iterates through all hyperedge pairs to form $G = (V,E)$, and then applies a black-box linear time \vc{} algorithm will take $O(\sum_{v \in \mathcal{V}} d_v^2 + |\mathcal{E}|^2)$-time where $d_v$ is the degree of node $v \in \mathcal{V}$. However, an implicit implementation of Pitt's \vc{} algorithm leads to a 2-approximation with a runtime of $O(\sum_{e \in E} |e|)$~\cite{veldt2022optimal}, which is linear in terms of the hypergraph size. We complement this result with a an even simpler randomized 2-approximation that corresponds to an implicit implementation of \textsc{NeighborCover} (Algorithm~\ref{alg:nccolor}). 

\begin{algorithm}[t]
	\caption{$\textsc{NeighborCoverColorEC}(G)$}
	\label{alg:nccolor}
	\begin{algorithmic}
		\State{\bfseries Input:} Edge-colored hypergraph $\mathcal{H} = (\mathcal{V}, \mathcal{E}, \ell, k)$
		\State {\bfseries Output:} Node label function $Y \colon V \rightarrow \{1,2, \hdots k\}$
		\State $\sigma = \textsc{WeightedShuffle}(\{w_1, w_2, \hdots , w_{|\mathcal{E}|}\})$
		\State For each $v \in V$ set $Y[u] = 0$
		\For{$e = e_{\sigma(1)}, e_{\sigma(2)}, \hdots, e_{\sigma(|\mathcal{E}|)}$}
		\If{$Y[u] \in \{\ell(e), 0\}$ for every $u \in e$}
		\For{$u \in e$}
		\State $Y[u] = \ell(e)$
		\EndFor
		\EndIf
		\EndFor
		\State Return $Y$
	\end{algorithmic}
\end{algorithm}
\begin{figure}[t]
	\centering
	\includegraphics[width = .8\linewidth]{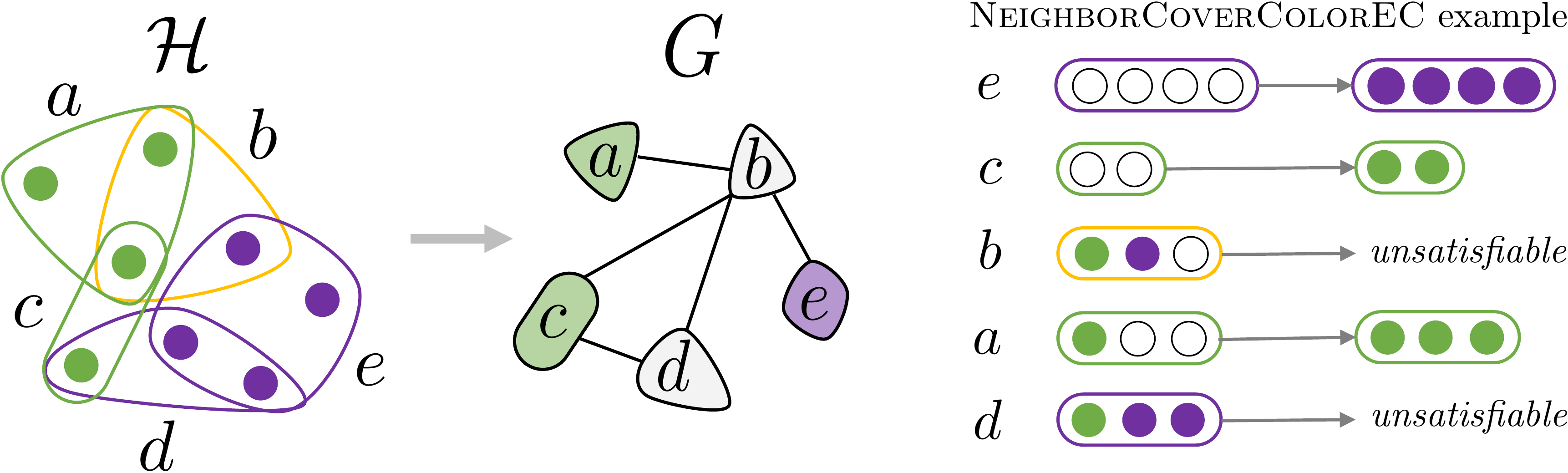}
	\caption{A small example of converting an edge-colored hypergraph $\mathcal{H}$ into \vc{} on a graph $G$ and implicitly applying \textsc{NeighborCover} to approximate \textsc{Colored Edge Clustering}. Hyperedge colors are given as input; node colors in $\mathcal{H}$ are the result of running the algorithm. Colored nodes in $G$ indicate maximal independent set nodes, and are colored to match hyperedge colors in $\mathcal{H}$. Gray nodes in $G$ are vertex cover nodes. Growing a maximal independent set in $G$ is equivalent to visiting hyperedges at random and \emph{satisfying} them when possible, meaning that all nodes are given the color of the hyperedge. If an edge is not satisfiable, the algorithm does nothing.}
	\label{fig:colorec}
\end{figure}

Finding a maximal independent set in $G$ is equivalent to finding a maximal set of satisfied  hyperedges. Using the hyperedge-deletion view of the objective, an unsatisfied hyperedge is a hyperedge that must be deleted. 
Following the basic strategy of \textsc{NeighborCover}, Algorithm~\ref{alg:nccolor} iterates through the hyperedges in $\mathcal{H}$ (i.e., nodes in $G$) and greedily adds them to the satisfied set (i.e., an independent set in $G$). The only reason to not add a hyperedge $e$ to the satisfied set is if an overlapping hyperedge of a different color (i.e., an adjacent node in $G$) was already satisfied in an earlier iteration. This means that at least one of the nodes $u \in e$ was already assigned a color $Y[u] \neq \ell(e)$. Therefore, in an iteration where we visit a hyperedge $e$, we simply need to check the current color assignment for each node in $e$, and give all these node color $\ell(e)$ if possible. The fact that this algorithm is a 2-approximation is a corollary of Theorem~\ref{thm:mainalg}. 
\begin{corollary}
	Algorithm~\ref{alg:nccolor} is a randomized 2-approximation for \textsc{Colored Edge Clustering}. Its runtime is $O(|\mathcal{E}| \log |\mathcal{E}| + \sum_{e \in \mathcal{E}} |e|)$ for weighted hypergraphs and $O(\sum_{e \in \mathcal{E}} |e|)$ for unweighted hypergraphs.
\end{corollary}
For the runtime analysis, note that it takes $O(\sum_{e \in \mathcal{E}} |e|)$ time to iterate through all the edges looking for satisfiable edges. The additional $O(|\mathcal{E}| \log |\mathcal{E}|)$ term for the weighted case comes from applying Algorithm~\ref{alg:randperm} to order edges.

	\section{Conclusions and Discussion}
	We have introduced a simple new approximation algorithm for \vc{} and have discussed its connections to related previous algorithms for clustering nodes, labeling edges, and finding maximal independent sets. This method leads to fast and simple approximation algorithms for certain edge-deletion problems that can be reduced to \vc{} in an approximation preserving way. One open direction is to explore other problems that are reducible to \vc{} which might also benefit from implicit implementations. There are in fact examples where applying our method implicitly \emph{does not} lead to runtime improvements over explicitly forming the reduced \vc{} instance. One example is a version of the STC edge-labeling problem that does not allow edge additions~\cite{sintos2014using}. This problem can be reduced to \vc{} and provides a lower bound for a variant of \cc{} called \textsc{Cluster Deletion}~\cite{gruttemeier2020relation,sintos2014using,veldt2022correlation}. 
	We were unable to develop faster approximation algorithms for either problem using implicit implementations of \textsc{NeighborCover}, as there does not appear to be a way to avoid iterating through all edges in the reduced \vc{} instance. 
	
	One disadvantage of our method is that the weighted version involves an $O(|V| \log |V|)$ time node sampling step, whereas several previous algorithms for \vc{} run in linear time even in the node-weighted case. An $O(|E|)$-time implementation for the weighted version of our algorithm would be a useful improvement, though this seems challenging. Another advantage of some other \vc{} algorithms is that they generalize easily to hypergraph \vc{}. Although \textsc{Pivot} can be viewed as a 3-approximation algorithm for \textsc{Vertex Cover} in a very restrictive type of 3-uniform hypergraph, generalizing \textsc{NeighborCover} to the general 3-uniform hypergraph \textsc{Vertex Cover} problem remains open.
	
	The connections highlighted in Section~\ref{sec:equiv} suggest several other compelling directions for future research. Our equivalence results show that a simple parallel version of the greedy MIS algorithm also approximates \vc{}.
	Although this is not the first parallel 2-approximation for \vc{}, nor the best in terms of the number of rounds, it is particularly simple and has the attractive feature that it simultaneously solves multiple problems. There do exist $O(\log \log n)$-round algorithms for finding a maximal independent set and for approximating \vc{}~\cite{ghaffari2018improved}, but a different approach is used for each problem. Furthermore, the first $O(\log \log n)$-round $(2+\varepsilon)$-approximation for \emph{weighted} \textsc{Vertex Cover} was developed separately and required yet a different approach~\cite{ghaffari2020massively}. One interesting direction for future research is to explore whether the approximation guarantee for \textsc{NeighborCover}, coupled with the fact that this algorithm applies to weighted \textsc{Vertex Cover}, can be used to further simplify, unify, and improve existing parallel algorithms for MIS, \textsc{Vertex Cover}, and \textsc{Correlation Clustering}. Another open question is to see whether we can leverage weighted \vc{} algorithms to develop faster approximation algorithms for weighted variants of \cc{}.
	
	\bibliographystyle{plain}
	\bibliography{vertex-cover-bib}

\begin{thebibliography}{10}

\bibitem{abu2006scalable}
Faisal~N Abu-Khzam, Michael~A Langston, Pushkar Shanbhag, and Christopher~T
  Symons.
\newblock Scalable parallel algorithms for fpt problems.
\newblock {\em Algorithmica}, 45(3):269--284, 2006.

\bibitem{ailon2008aggregating}
Nir Ailon, Moses Charikar, and Alantha Newman.
\newblock Aggregating inconsistent information: ranking and clustering.
\newblock {\em Journal of the ACM (JACM)}, 55(5):1--27, 2008.

\bibitem{alon1986fast}
Noga Alon, L{\'a}szl{\'o} Babai, and Alon Itai.
\newblock A fast and simple randomized parallel algorithm for the maximal
  independent set problem.
\newblock {\em Journal of algorithms}, 7(4):567--583, 1986.

\bibitem{amburg2020clustering}
Ilya Amburg, Nate Veldt, and Austin Benson.
\newblock Clustering in graphs and hypergraphs with categorical edge labels.
\newblock In {\em Proceedings of The Web Conference 2020}, pages 706--717,
  2020.

\bibitem{avis2007list}
David Avis and Tomokazu Imamura.
\newblock A list heuristic for vertex cover.
\newblock {\em Operations research letters}, 35(2):201--204, 2007.

\bibitem{BansalBlumChawla2004}
Nikhil Bansal, Avrim Blum, and Shuchi Chawla.
\newblock Correlation clustering.
\newblock {\em Machine Learning}, 56:89--113, 2004.

\bibitem{bar1981linear}
Reuven Bar-Yehuda and Shimon Even.
\newblock A linear-time approximation algorithm for the weighted vertex cover
  problem.
\newblock {\em Journal of Algorithms}, 2(2):198--203, 1981.

\bibitem{bar1985local}
Reuven Bar-Yehuda and Shimon Even.
\newblock A local-ratio theorm for approximating the weighted vertex cover
  problem.
\newblock {\em Annals of Discrete Mathematics}, 25:27--46, 1985.

\bibitem{behnezhad2022almost}
Soheil Behnezhad, Moses Charikar, Weiyun Ma, and Li-Yang Tan.
\newblock Almost 3-approximate correlation clustering in constant rounds.
\newblock {\em arXiv preprint arXiv:2205.03710}, 2022.

\bibitem{bennett2016note}
Patrick Bennett and Tom Bohman.
\newblock A note on the random greedy independent set algorithm.
\newblock {\em Random Structures \& Algorithms}, 49(3):479--502, 2016.

\bibitem{blelloch2012greedy}
Guy~E Blelloch, Jeremy~T Fineman, and Julian Shun.
\newblock Greedy sequential maximal independent set and matching are parallel
  on average.
\newblock In {\em Proceedings of the twenty-fourth annual ACM symposium on
  Parallelism in algorithms and architectures}, pages 308--317, 2012.

\bibitem{ChawlaMakarychevSchrammEtAl2015}
Shuchi Chawla, Konstantin Makarychev, Tselil Schramm, and Grigory Yaroslavtsev.
\newblock Near optimal {LP} rounding algorithm for correlation clustering on
  complete and complete k-partite graphs.
\newblock In {\em Proceedings of the Forty-Seventh Annual ACM on Symposium on
  Theory of Computing}, pages 219--228. ACM, 2015.

\bibitem{chen2008kernel}
Jianer Chen.
\newblock {\em Vertex Cover Kernelization}, pages 1003--1005.
\newblock Springer US, Boston, MA, 2008.

\bibitem{chen2001vertex}
Jianer Chen, Iyad~A Kanj, and Weijia Jia.
\newblock Vertex cover: further observations and further improvements.
\newblock {\em Journal of Algorithms}, 41(2):280--301, 2001.

\bibitem{chen2006improved}
Jianer Chen, Iyad~A Kanj, and Ge~Xia.
\newblock Improved parameterized upper bounds for vertex cover.
\newblock In {\em International symposium on mathematical foundations of
  computer science}, pages 238--249. Springer, 2006.

\bibitem{cohen2021correlation}
Vincent Cohen-Addad, Silvio Lattanzi, Slobodan Mitrovi{\'c}, Ashkan
  Norouzi-Fard, Nikos Parotsidis, and Jakub Tarnawski.
\newblock Correlation clustering in constant many parallel rounds.
\newblock In {\em International Conference on Machine Learning}, pages
  2069--2078. PMLR, 2021.

\bibitem{cohen2022correlation}
Vincent Cohen-Addad, Euiwoong Lee, and Alantha Newman.
\newblock Correlation clustering with sherali-adams.
\newblock {\em arXiv preprint arXiv:2207.10889}, 2022.

\bibitem{coppersmith1987parallel}
Don Coppersmith, Prabhakar Raghavan, and Martin Tompa.
\newblock Parallel graph algorithms that are efficient on average.
\newblock In {\em 28th Annual Symposium on Foundations of Computer Science
  (sfcs 1987)}, pages 260--269. IEEE, 1987.

\bibitem{delbot2008better}
Fran{\c{c}}ois Delbot and Christian Laforest.
\newblock A better list heuristic for vertex cover.
\newblock {\em Information Processing Letters}, 107(3-4):125--127, 2008.

\bibitem{dinur2005hardness}
Irit Dinur and Samuel Safra.
\newblock On the hardness of approximating minimum vertex cover.
\newblock {\em Annals of mathematics}, pages 439--485, 2005.

\bibitem{easley2010networks}
David Easley and Jon Kleinberg.
\newblock {\em Networks, crowds, and markets}, volume~8.
\newblock Cambridge university press Cambridge, 2010.

\bibitem{feige1996interactive}
Uriel Feige, Shafi Goldwasser, L{\'a}szl{\'o} Lov{\'a}sz, Shmuel Safra, and
  Mario Szegedy.
\newblock Interactive proofs and the hardness of approximating cliques.
\newblock {\em Journal of the ACM (JACM)}, 43(2):268--292, 1996.

\bibitem{fischer2019tight}
Manuela Fischer and Andreas Noever.
\newblock Tight analysis of parallel randomized greedy mis.
\newblock {\em ACM Transactions on Algorithms (TALG)}, 16(1):1--13, 2019.

\bibitem{Garey:1990:CIG:574848}
Michael~R. Garey and David~S. Johnson.
\newblock {\em Computers and Intractability; A Guide to the Theory of
  NP-Completeness}.
\newblock W. H. Freeman \& Co., New York, NY, USA, 1990.

\bibitem{garg2004multiway}
Naveen Garg, Vijay~V Vazirani, and Mihalis Yannakakis.
\newblock Multiway cuts in node weighted graphs.
\newblock {\em Journal of Algorithms}, 50(1):49--61, 2004.

\bibitem{garg2016raising}
Shivam Garg and Geevarghese Philip.
\newblock Raising the bar for vertex cover: Fixed-parameter tractability above
  a higher guarantee.
\newblock In {\em Proceedings of the Twenty-Seventh Annual ACM-SIAM Symposium
  on Discrete Algorithms}, pages 1152--1166. SIAM, 2016.

\bibitem{ghaffari2018improved}
Mohsen Ghaffari, Themis Gouleakis, Christian Konrad, Slobodan Mitrovi{\'c}, and
  Ronitt Rubinfeld.
\newblock Improved massively parallel computation algorithms for mis, matching,
  and vertex cover.
\newblock In {\em Proceedings of the 2018 ACM Symposium on Principles of
  Distributed Computing}, pages 129--138, 2018.

\bibitem{ghaffari2020massively}
Mohsen Ghaffari, Ce~Jin, and Daan Nilis.
\newblock A massively parallel algorithm for minimum weight vertex cover.
\newblock In {\em Proceedings of the 32nd ACM Symposium on Parallelism in
  Algorithms and Architectures}, pages 259--268, 2020.

\bibitem{gottlieb2014efficient}
Lee-Ad Gottlieb, Aryeh Kontorovich, and Robert Krauthgamer.
\newblock Efficient classification for metric data.
\newblock {\em IEEE Transactions on Information Theory}, 60(9):5750--5759,
  2014.

\bibitem{gruttemeier2022parameterized}
Niels Gr{\"u}ttemeier.
\newblock {\em Parameterized Algorithmics for Graph-Based Data Analysis}.
\newblock PhD thesis, Philipps-Universit{\"a}t Marburg, 2022.

\bibitem{gruttemeier2020relation}
Niels Gr{\"u}ttemeier and Christian Komusiewicz.
\newblock On the relation of strong triadic closure and cluster deletion.
\newblock {\em Algorithmica}, 82(4):853--880, 2020.

\bibitem{halperin2002improved}
Eran Halperin.
\newblock Improved approximation algorithms for the vertex cover problem in
  graphs and hypergraphs.
\newblock {\em SIAM Journal on Computing}, 31(5):1608--1623, 2002.

\bibitem{karakostas2009better}
George Karakostas.
\newblock A better approximation ratio for the vertex cover problem.
\newblock {\em ACM Trans. Algorithms}, 5(4), nov 2009.

\bibitem{karp1972reducibility}
Richard~M Karp.
\newblock Reducibility among combinatorial problems.
\newblock In {\em Complexity of computer computations}, pages 85--103.
  Springer, 1972.

\bibitem{karp1985fast}
Richard~M Karp and Avi Wigderson.
\newblock A fast parallel algorithm for the maximal independent set problem.
\newblock {\em Journal of the ACM (JACM)}, 32(4):762--773, 1985.

\bibitem{kenkre2017approximability}
Sreyash Kenkre, Vinayaka Pandit, Manish Purohit, and Rishi Saket.
\newblock On the approximability of digraph ordering.
\newblock {\em Algorithmica}, 78(4):1182--1205, 2017.

\bibitem{khot2008vertex}
Subhash Khot and Oded Regev.
\newblock Vertex cover might be hard to approximate to within 2- $\varepsilon$.
\newblock {\em Journal of Computer and System Sciences}, 74(3):335--349, 2008.

\bibitem{klein2016approximability}
Nathan Klein.
\newblock On the approximability of dag edge deletion.
\newblock Bachelor's thesis, Oberlin College, 2016.

\bibitem{koufogiannakis2009distributed}
Christos Koufogiannakis and Neal~E. Young.
\newblock Distributed and parallel algorithms for weighted vertex cover and
  other covering problems.
\newblock In {\em Proceedings of the 28th ACM Symposium on Principles of
  Distributed Computing}, PODC '09, pages 171--179, New York, NY, USA, 2009.
  Association for Computing Machinery.

\bibitem{kratsch2018randomized}
Stefan Kratsch.
\newblock A randomized polynomial kernelization for vertex cover with a smaller
  parameter.
\newblock {\em SIAM Journal on Discrete Mathematics}, 32(3):1806--1839, 2018.

\bibitem{lampis2008algorithmic}
Michael Lampis, Georgia Kaouri, and Valia Mitsou.
\newblock On the algorithmic effectiveness of digraph decompositions and
  complexity measures.
\newblock In {\em International Symposium on Algorithms and Computation}, pages
  220--231. Springer, 2008.

\bibitem{luby1986simple}
Michael Luby.
\newblock A simple parallel algorithm for the maximal independent set problem.
\newblock {\em SIAM journal on computing}, 15(4):1036--1053, 1986.

\bibitem{nuendorf2020}
B{\'e}la Nuendorf.
\newblock On strong triadic closure with edge insertion.
\newblock Bachelor's thesis, Philipps-Universit{\"a}t Marburg, 2020.

\bibitem{paik1994deleting}
Doowon Paik, Sudhakar Reddy, and Sartaj Sahni.
\newblock Deleting vertices to bound path length.
\newblock {\em IEEE transactions on computers}, 43(9):1091--1096, 1994.

\bibitem{papadimitriou1998combinatorial}
Christos~H Papadimitriou and Kenneth Steiglitz.
\newblock {\em Combinatorial optimization: algorithms and complexity}.
\newblock Courier Corporation, 1998.

\bibitem{park2001effectiveness}
Kihong Park and Heejo Lee.
\newblock On the effectiveness of route-based packet filtering for distributed
  dos attack prevention in power-law internets.
\newblock {\em ACM SIGCOMM computer communication review}, 31(4):15--26, 2001.

\bibitem{pitt1985simple}
Leonard~Brian Pitt.
\newblock {\em A simple probabilistic approximation algorithm for vertex
  cover}.
\newblock Yale University, Department of Computer Science, 1985.

\bibitem{savage1982depth}
Carla Savage.
\newblock Depth-first search and the vertex cover problem.
\newblock {\em Information processing letters}, 14(5):233--235, 1982.

\bibitem{sintos2014using}
Stavros Sintos and Panayiotis Tsaparas.
\newblock Using strong triadic closure to characterize ties in social networks.
\newblock In {\em Proceedings of the 20th ACM SIGKDD international conference
  on Knowledge discovery and data mining}, KDD '14, pages 1466--1475, 2014.

\bibitem{veldt2022correlation}
Nate Veldt.
\newblock Correlation clustering via strong triadic closure labeling: Fast
  approximation algorithms and practical lower bounds.
\newblock In {\em Proceedings of the 2022 International Conference on Machine
  Learning}, ICML '22, 2022.

\bibitem{veldt2022optimal}
Nate Veldt.
\newblock Optimal {LP} rounding and fast combinatorial algorithms for
  clustering edge-colored hypergraphs.
\newblock {\em arXiv preprint arXiv:2208.06506}, 2022.

\bibitem{wong1980efficient}
C.~K. Wong and M.~C. Easton.
\newblock An efficient method for weighted sampling without replacement.
\newblock {\em SIAM Journal on Computing}, 9(1):111--113, 1980.

\bibitem{zuckerman2006linear}
David Zuckerman.
\newblock Linear degree extractors and the inapproximability of max clique and
  chromatic number.
\newblock In {\em Proceedings of the thirty-eighth annual ACM symposium on
  Theory of computing}, pages 681--690, 2006.

\end{thebibliography}
\end{document}